\documentclass[letterpaper, 10 pt, conference]{ieeeconf}
\IEEEoverridecommandlockouts
\usepackage{graphics,epsfig,rotating,setspace,latexsym,amsmath,epsf,amssymb,bm,theorem}
\usepackage{cite,enumerate,color}
\usepackage{authblk,subcaption}

\newtheorem{theorem}{Theorem}

\newtheorem{lemma}{Lemma}
\newtheorem{proposition}{Proposition}

\newcommand{\sqrtD}[1]{\sqrt{#1}\,}
\newcommand{\mb}{\mathbb}

\newcommand{\sgn}{\text{sign}}
\DeclareMathOperator{\arccot}{arccot}

\title{\LARGE \bf Asymptotic Performance Analysis of Majority Sentiment Detection in Online Social Networks}
\author{Tian Tong$^{1}$ and Rohit Negi$^{2}$
\thanks{*This research was partially supported by NSF awards CCF1422193 and CNS1218823.}
\thanks{Electrical and Computer Engineering Department, Carnegie Mellon University, {$^{1}$\tt\small ttong1@andrew.cmu.edu}, {$^{2}$\tt\small negi@ece.cmu.edu} 
}}

\begin{document}

\maketitle
\thispagestyle{empty}
\pagestyle{empty}

\begin{abstract}
We analyze the problem of majority sentiment detection in Online Social Networks (OSN), and relate the detection error probability to the 
underlying graph of the OSN. Modeling the underlying social network as an Ising Markov random field prior based on a given graph, we show that in the case of the empty graph (independent sentiments) and the chain graph, the detection is always inaccurate, even when the number of users grow to infinity. In the case of the complete graph, the detection is inaccurate if the connection strength is below a certain critical value, while it is asymptotically accurate if the strength is above that critical value, which is analogous to the phase transition phenomenon in statistical physics.
\end{abstract}

\section{Introduction}
\label{sec:introduction}
Online social networks (OSN), such as Facebook and Twitter\cite{howard2013democracy}, have a significant influence on people and  society. The massive data embedded in these networks have turned OSNs into a gold mine for politicians, economists, and sociologists alike to collect, analyze, and understand the views of people. Therefore, detecting and analyzing the sentiments of OSN users is of great interest in recent machine learning and sociology research \cite{kouloumpis2011twitter,thelwall2011sentiment}.

We focus on the problem of {\em majority sentiment detection}, also known as the vote detection, where the majority sentiment is estimated based on {\em noisy measurements} of user sentiments. (Majority sentiment is the one among two binary sentiments, such as `approve'/`disapprove', which predominates among the users.) Related research abounds under various topics such as public opinion studies \cite{lippmann1946public}, voting theory \cite{coughlin1992probabilistic}, and opinion mining \cite{pang2008opinion}. The basic assumptions are that users (or members) in the network are connected by some relationships, such as the friend relationship in Facebook, or follower/followee relationship in Twitter, and that two connected users are more probable to share the same sentiment, since the network typically signifies affinity of opinion. At the same time, automated language processing tools that measure individual sentiments from posts or tweets suffer from noise, due to short size of the text or the inability to recognize sarcasm \cite{negi2014latent}.

In this paper, we attempt to answer an interesting question: how does the network (graph\footnote{We use the terms `graph of OSN' and `network' interchangeably.} of OSN) influence the error performance of (automated) majority sentiment detection? In particular, we wish to investigate whether such detection is asymptotically accurate, i.e., whether the error probability becomes arbitrarily small as the network size, in terms of number of users, grows.
As we will show, the error performance is strongly related to the graph of the OSN. It involves two levels of influence: the graph structure and the strength of the connections.  

To analyze the performance of majority sentiment detection under various network topologies, we model the network by an Ising Markov Random Field (MRF) model \cite{binder2001ising}. This model was first introduced in statistical physics to interpret the paramagnetic-ferromagnetic phase transition phenomenon.
We first provide general upper and lower bounds on the asymptotic detection error probability. Next, we consider special cases of networks to illustrate the phase-transition-like phenomenon in the error probability behavior, where the detection is either inaccurate or is asymptotically accurate.   Specifically, we  show that in
the case of empty graph and chain graph, both of which are weakly connected networks, {\em the detection is always inaccurate, even when the number of users grows to infinity}.
This result appears to be of interest in its own right, since one may naively expect accurate performance in the limit of infinite users.
On the other hand, in the complete graph (with standard scaling down of the connection strength),  {\em there exists a critical value for connection strength} 
(analogous to the critical temperature in statistical physics), below which the detection is inaccurate while above which the detection is asymptotically accurate.

In Section \ref{sec:model}, we introduce the Ising model of OSN sentiments and analyze majority
sentiment detection in the independent sentiment case. In Section \ref{sec:mainTheory}, we obtain bounds on the asymptotic sentiment detection error probability for arbitrary graphs of the Ising model. In Section \ref{sec:networkExp}, we consider special cases of graphs to calculate these bounds. Section \ref{sec:num} shows numerical results while Section \ref{sec:conclusion} concludes the paper.

\section{System Model}
\label{sec:model}
\begin{figure}[t]
\centerline{\includegraphics[width=\linewidth]{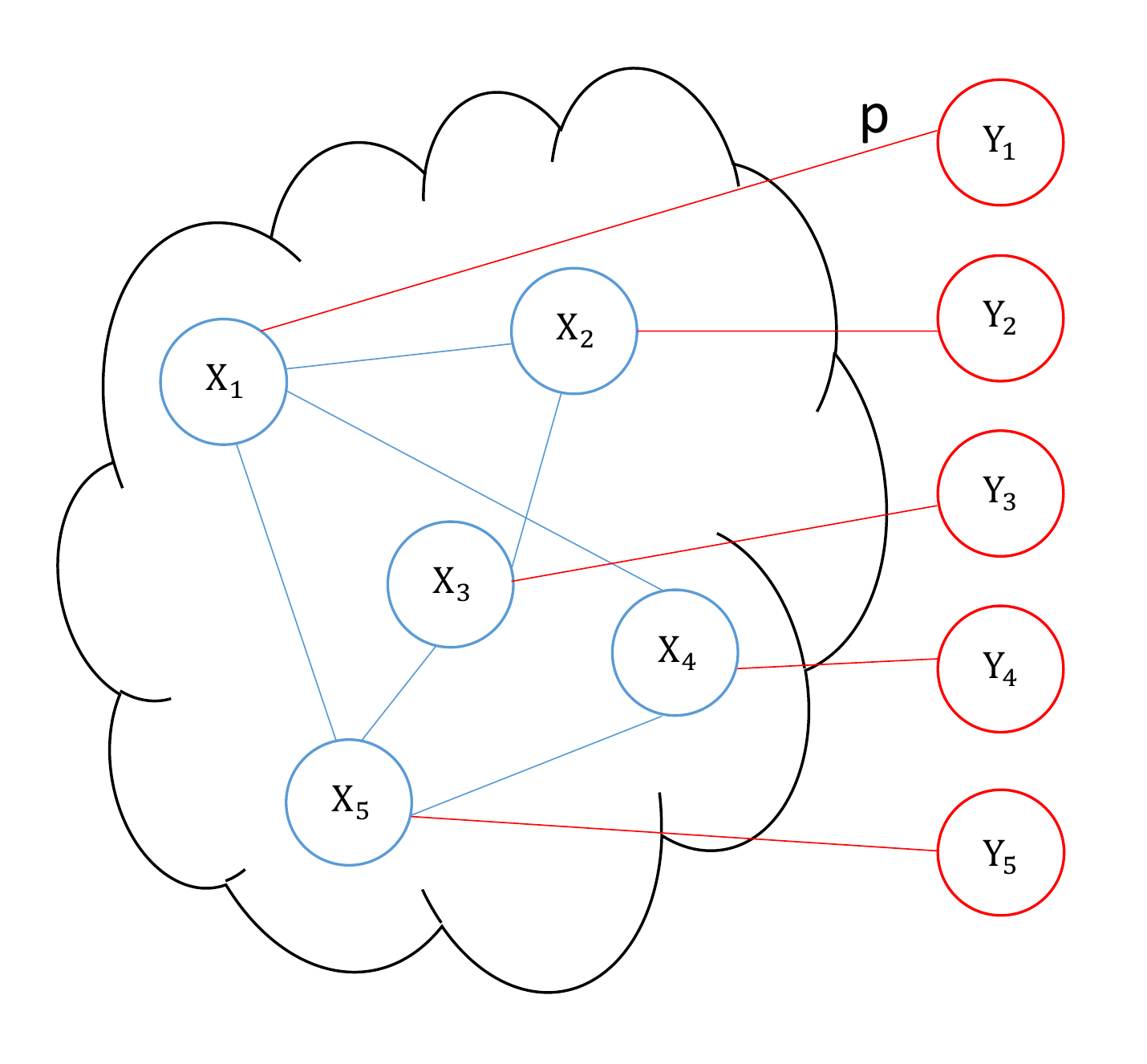}}
\caption{Markov Random Field model of sentiment detection.}
\label{fig_model}
\vspace*{-0.2cm}
\end{figure}

The social network structure is modeled as an undirected graph, as shown in Fig.~\ref{fig_model}. Let $\mathbf{X} = (X_1, \dots, X_n)^T \in \{-1, 1\}^n$ denote the vector of binary sentiments of $n$ members, where $1$ or $-1$ denote positive or negative sentiments respectively. These sentiments are unknown to observers. What is observed is a noisy measurement of $\mathbf{X}$, called $\mathbf{Y}$. $\mathbf{Y} = (Y_1, \dots, Y_n)^T \in \{-1, 1\}^n$ is modeled as conditionally independent binary measurements of $\mathbf{X}$ each with cross-over probability $p$, i.e., the output of a binary symmetric channel with input $\mathbf{X}$. Without loss of generality, we assume that $p < \frac{1}{2}$. This model was first introduced in \cite{negi2014latent} for the `latent sentiment detection' problem, with more details about the probability mass function of $\bf{X}$ following in Section~\ref{section_networkModel}. 

The majority sentiment is defined as
\begin{align}
m = \sgn({\bf 1}^T {\bf X}),
\end{align}
where ${\bf 1}$ denotes the vector of all ones. We assume that $n$ is odd to avoid trivial ambiguity.

We will use the majority vote detector as our estimator for the majority sentiment:
\begin{align}
\label{eq_majorDetector}
\hat{m} = \sgn({\bf 1}^T {\bf Y}).
\end{align}
This detector estimates the majority sentiment as the majority of noisy measurements. In general, it is not the optimal Maximum Aposteriori Probability
(MAP) detector. However, it will be sufficient to illustrate the key insights of this paper, promised in Section \ref{sec:introduction}, as will be shown below. 

The detection error probability (equivalently, classification error probability of the two sentiment class problem) is 
\begin{align}
\label{eq_errorProb}
P^{(n)}_e = \mb{P}(m \neq \hat{m}).
\end{align}
We wish to investigate whether
the majority sentiment detector is asymptotically accurate, i.e., whether $P^{(n)}_e$ becomes arbitrarily small when $n$ is sufficiently large. 

In the paper, we denote $\overline{X_n} = \frac{1}{n}\sum_{i=1}^n X_i$ as the sample average, $X_{\lim}$ as a variable distributed as the limiting distribution of $\sqrtD{n}\overline{X_n}$, and $\xrightarrow{d}$ as convergence in distribution.

\subsection{I.I.D Case}
\label{section_empty}
Before discussing  interesting networks, let us first consider the simplest case: $X_i$s are independent and identical distributed (i.i.d.) random variables, taking values $-1$ or $1$, each with probability $1/2$. In this example, the members are not connected, or in other words the network is an empty graph. This is the case in contemporary vote detection schemes where voters are assumed to independently vote, and the  underlying network is not modeled \cite{coughlin1992probabilistic}. For this case, the detector we adopt in (\ref{eq_majorDetector}) is actually the optimal MAP detector, so that there is no loss of error performance. 

For this case, we can derive the exact value of the asymptotic error probability, as shown in Proposition~\ref{pp_Empty-Graph}. 

\begin{proposition}
\label{pp_Empty-Graph}
In the i.i.d. (empty graph) case, $$\lim_{n\to\infty} P_e = \frac{2}{\pi}\arcsin\sqrt{p} > 0.$$
\end{proposition}

\begin{proof}
Since $(X_i, Y_i)$s are i.i.d. random variables, with means zero, and variances $\mb{E}[X_i^2] = \mb{E}[Y_i^2] = 1$, $\mb{E}[X_i Y_i] = 1-2p$, the multidimensional central limit theorem tells that:
\begin{align}
\label{eq_emptyCLT}
\left(\begin{matrix} \sqrtD{n}\overline{X_n} \\ \sqrtD{n}\overline{Y_n}\end{matrix}\right)\xrightarrow{d} \left(\begin{matrix} X_{\lim} \\ Y_{\lim}\end{matrix}\right) \sim N\left(\left(\begin{matrix} 0 \\ 0\end{matrix}\right), \left(\begin{matrix} 1 & 1-2p \\ 1-2p & 1\end{matrix}\right)\right).
\end{align}
Since by definition (\ref{eq_errorProb}), $P_e^{(n)}= \mb{P}(\sqrtD{n}\overline{X_n} \sqrtD{n}\overline{Y_n} < 0)$, we can calculate its limit from (\ref{eq_emptyCLT}) as 
\begin{align}
\lim_{n\to\infty} P^{(n)}_e &= \mb{P}(X_{\lim} Y_{\lim} < 0) \nonumber = \frac{2}{\pi}\arcsin \sqrt{p}, 
\end{align}
where the limit exists thanks to the convergence in distribution (\ref{eq_emptyCLT}). 
\end{proof}

In this special case, the detection error probability tends to some positive constant, and thus never reduces to $0$, even with infinite
number of sentiments. Inherently, this inaccuracy results from the absence of connections. {\em This result is perhaps counterintuitive, since one may have expected that the majority sentiment can be detected accurately when the user size tends to infinity, as typically happens in the case of single parameter estimation.}

\subsection{Network Model}
\label{section_networkModel}
The network model characterizes the probability mass function of ${\bf X}$. In general, it could be any prior implying that members with connections are more probable to have the same sentiment. In this paper following \cite{negi2014latent}, we adopt a homogeneous Ising MRF prior as
\begin{align}
\label{eq_isingPrior}
p({\bf x}) = \frac{\exp(\theta {\bf x}^TA {\bf x})}{Z_n(\theta)}.
\end{align}
Here $A$ denotes the symmetric graph adjacency matrix, with $A_{ij}=0/1$ denoting absence/presence of an edge respectively. $\theta > 0$ is called the inverse temperature parameter. We remark that $\theta$ characterizes the connection strength in the network, namely, connected members are more probable to share the same sentiment in case of larger $\theta$. The normalizer $Z_n(\theta) = \sum_{\mathbf{x}\in\{-1,1\}^n}\exp(\theta {\bf x}^T A {\bf x})$ is called the partition function. 

The joint distribution of $({\bf x}, {\bf y})$ can be written as another Ising MRF as 
\begin{align}
p(\mathbf{x}, \mathbf{y}) = \frac{\exp(\theta {\bf x}^TA {\bf y} + \varepsilon {\bf y}^T {\bf x})}{Z_n(\theta)(2\cosh \varepsilon)^n}.
\end{align}
Here $\varepsilon$ is defined by $p = \frac{\exp(-\varepsilon)}{\exp(\varepsilon) + \exp(-\varepsilon)}$. In the next section, we analyze the asymptotic 
error probability under the above network model. For this purpose, we assume that as $n$ grows, there is a given sequence of graphs of $n$ vertices that models the OSN as (\ref{eq_isingPrior}). 

\section{Error Probability under Network Model}
\label{sec:mainTheory}
In this section, we will first derive an upper bound for the detection error probability in Theorem~\ref{thm_Hoeffd}, and two exact asymptotic results in Theorems~\ref{thm_Q} and \ref{thm_lim}, assuming a given sequence of graph adjacency matrices $A$. After that, we will show that the asymptotic performance of the detection error probability is related to the concentration behavior of $\sqrtD{n}\overline{X_n}$ in Theorem~\ref{thm_tt}.
\begin{theorem}
\label{thm_Hoeffd}
$$P_e^{(n)} \le \mb{E}\left[\exp\left(\frac{-(1-2p)^2 }{8(1-p)^2}(\sqrtD{n}\overline{X_n})^2\right)\right].$$ Here the expectation is taken over ${\bf X}$.
\end{theorem}
\begin{proof}
Let $Z_i = Y_i - (1-2p)X_i$. Since $Z_i$s are conditionally independent given ${\bf X}$, with $\mb{E}[Z_i|{\bf X}] = 0$ and $\lvert Z_i \rvert \le 2(1-p)$, Hoeffding's inequality tells that average $\overline{Z_n}$ satisfies $\mb{P}(\sqrtD{n}\overline{Z_n} > \epsilon \mid {\bf X}) \le \exp\left(\frac{-\epsilon^2}{8(1-p)^2}\right)$ and $\mb{P}(\sqrtD{n}\overline{Z_n} < -\epsilon \mid {\bf X}) \le \exp\left(\frac{-\epsilon^2}{8(1-p)^2}\right)$ for any $\epsilon > 0$. By definition (\ref{eq_errorProb}):
\begin{align}
\label{eq_PeFormula}
P_e^{(n)} &= \mb{P}(\sqrtD{n} \overline{X_n} \sqrtD{n} \overline{Y_n} < 0) \nonumber\\
&= \mb{P}(\sqrtD{n}\overline{X_n} \sqrtD{n}\overline{Z_n} < -(1-2p)(\sqrtD{n}\overline{X_n})^2) \nonumber \\
&= \mb{E}\left[ \mb{P}(\sqrtD{n}\overline{X_n} \sqrtD{n}\overline{Z_n} < -(1-2p)(\sqrtD{n}\overline{X_n})^2 \mid {\bf X}) \right]\\
& \le \mb{E}\left[\exp\left(\frac{-(1-2p)^2 }{8(1-p)^2}(\sqrtD{n}\overline{X_n})^2\right)\right]. \nonumber
\end{align}
\end{proof}

\begin{theorem}
\label{thm_Q}
$$\liminf_{n\to\infty}P_e^{(n)} = \liminf_{n\to\infty}\mb{E} \left[Q\left(\frac{(1-2p)}{\sqrt{4p(1-p)}}\left\lvert \sqrtD{n} \overline{X_n} \right\rvert \right)\right],$$ where $Q(\cdot)$ is the tail probability of standard normal distribution: $Q(x) = \frac{1}{\sqrt{2\pi}}\int_x^\infty \exp(-t^2/2)dt$.
\end{theorem}

\begin{theorem}
\label{thm_lim}
\begin{enumerate}[(a)]
\item If $\sqrtD{n}\overline{X_n} \xrightarrow{d} \Phi$, where $\Phi$ is a distribution, then
 $$\lim_{n\to\infty}P_e^{(n)} = \int_{-\infty}^\infty Q\left(\frac{(1-2p)}{\sqrt{4p(1-p)}}\left\lvert x \right\rvert \right) \Phi(dx).\footnote{This denotes Lebesgue integral with respect to probability measure $\Phi$.}$$ 

\item Specifically, if $\sqrtD{n}\overline{X_n} \xrightarrow{d} N(0, \sigma^2)$, then $$\lim_{n\to\infty}P_e^{(n)}= \frac{1}{\pi} \arccot\left(\frac{(1-2p)}{\sqrt{4p(1-p)}}\sigma\right) > 0.$$
\end{enumerate}
\end{theorem}

To prove Theorems~\ref{thm_Q} and \ref{thm_lim}, we first derive a central limit theorem type result in Lemma~\ref{lemma_CLT}. This will follow from the conditional independence of $Y_i$s given ${\bf X}$, and reveals that $\sqrtD{n}(\overline{Y_n} - (1-2p)\overline{X_n})$ tends to a normal distribution  conditioned on ${\bf X}$; besides that, when $\sqrt{n}\overline{X_n}$ converges in distribution, the joint limiting distribution of $\sqrt{n}(\overline{X_n}, \overline{Y_n})$ can also be obtained.
\begin{lemma}
\label{lemma_CLT} 
\begin{enumerate}[(a)]
\item Conditional convergence: For all ${\bf X}$, $$\sqrt{n}(\overline{Y_n} - (1-2p)\overline{X_n}) \mid {\bf X} \xrightarrow{d} N(0, 4p(1-p)).$$ 
\item Unconditional convergence: $$\sqrt{n}(\overline{Y_n} - (1-2p)\overline{X_n}) \xrightarrow{d} N(0, 4p(1-p)).$$
\item Joint convergence: If $\sqrtD{n}\overline{X_n} \xrightarrow{d} \Phi$, then $$\sqrt{n}\left(\overline{X_n}, \overline{Y_n} - (1-2p)\overline{X_n}\right) \xrightarrow{d} \left(\Phi, N(0, 4p(1-p))\right),$$ where the two limiting distributions are independent.
\end{enumerate} 
\end{lemma}

\begin{proof}
Let $Z_i = Y_i - (1-2p)X_i$. By L\'{e}vy's continuity theorem \cite{dudley2002real}, it is equivalent to prove the pointwise convergence of characteristic functions:
\begin{align*}
\lim_{n\to\infty} &\mb{E}[\exp(j\beta\sqrtD{n}\overline{Z_n}) | {\bf X}] = \exp(-4p(1-p)\beta^2/2), \\ 
\lim_{n\to\infty} &\mb{E}\left[\exp(j\beta\sqrtD{n}\overline{Z_n})\right] = \exp(-4p(1-p)\beta^2/2), \\ \lim_{n\to\infty} &\mb{E}\left[\exp(j\omega\sqrtD{n}\overline{X_n} + j\beta\sqrtD{n}\overline{Z_n})\right] = 
\\& \qquad \qquad \phi(\omega)\exp(-4p(1-p)\beta^2/2), 
\end{align*}
where $\phi(\cdot)$ denotes the characteristic function of $\Phi$. 

For part (a), since $Z_i$s are conditionally independent given ${\bf X}$, with $\mb{E}[Z_i| {\bf X}] = 0$, $\mb{E}[Z_i^2 | {\bf X}] = 4p(1-p)$, and $|Z_i| \le 2(1-p)$, the Lindeberg condition: $\forall \epsilon > 0$, {\small $$\lim_{n\to\infty}\frac{1}{n 4p(1-p)}\sum_{i=1}^n \mb{E}\left[Z^2_i I\left(|Z_i| \ge \epsilon \sqrt{n 4p(1-p)}\right) \mid {\bf X} \right] = 0$$}
is satisfied, where $I(\cdot)$ is the indicator function. Lindeberg-Feller central limit theorem \cite{dudley2002real} tells that: 
\begin{align}
\label{eq_CLTa}
\lim_{n\to\infty} \mb{E}[\exp(j\beta\sqrtD{n}\overline{Z_n})\mid {\bf X}] = \exp(-4p(1-p)\beta^2/2),
\end{align}
where RHS of (\ref{eq_CLTa}) is the characteristic function of $N(0, 4p(1-p))$.

For part (b), extend the result in (a) to the unconditional version. Notice that $\lvert \mb{E}[\exp(j\beta\sqrtD{n}\overline{Z_n})\mid \mathbf{X}] \rvert \le 1$, thus
\begin{align*}
\lim_{n\to\infty}&\mb{E}[\exp(j\beta\sqrtD{n}\overline{Z_n}] = \lim_{n\to \infty}\mb{E}\left[\mb{E}[\exp(j\beta\sqrtD{n}\overline{Z_n})\mid \mathbf{X}]\right] \\
&= \mb{E}\left[\lim_{n\to\infty}\mb{E}[\exp(j\beta\sqrtD{n}\overline{Z_n})\mid \mathbf{X}]\right] \\
&= \exp(-4p(1-p)\beta^2/2),
\end{align*}
where the expectation and limit are exchanged by Lebesgue's dominated convergence theorem \cite{dudley2002real}.

For part (c), define $\Delta_n({\bf X}) = \mb{E}[\exp(j\beta\sqrtD{n}\overline{Z_n})\mid \mathbf{X}] - \exp(-4p(1-p)\beta^2/2)$. By equation (\ref{eq_CLTa}), $\lim_{n\to\infty} \Delta_n({\bf X}) = 0, \forall {\bf X}$. So,
\begin{align}
\lim_{n\to\infty} &\mb{E}\left[\exp(j\omega\sqrtD{n}\overline{X_n} + j\beta\sqrtD{n}\overline{Z_n})\right] \nonumber\\
&= \lim_{n\to \infty}\mb{E}\left[\exp(j\omega\sqrtD{n}\overline{X_n})\mb{E}[\exp(j\beta\sqrtD{n}\overline{Z_n})\mid \mathbf{X}]\right] \nonumber \\
&= \lim_{n\to \infty}\mb{E}\left[\exp(j\omega\sqrtD{n}\overline{X_n})\exp(-4p(1-p)\beta^2/2)\right] \nonumber \\
& \qquad + \lim_{n\to \infty}\mb{E}[\exp(j\omega\sqrtD{n}\overline{X_n}) \Delta_n(\mathbf{X})] \label{eq_Lbgdelta}\\
&= \phi(\omega)\exp(-4p(1-p)\beta^2/2).\nonumber
\end{align}
In equation (\ref{eq_Lbgdelta}), since $\lvert \Delta_n({\bf X})\rvert \le 2$ is bounded, $\lim_{n\to \infty}\mb{E}[\exp(j\omega\sqrtD{n}\overline{X_n}) \Delta_n(\mathbf{X})] = 0$ by Lebesgue's dominated convergence theorem \cite{dudley2002real}. 
\end{proof}

Now we use Lemma~\ref{lemma_CLT} to prove Theorems~\ref{thm_Q} and \ref{thm_lim}. 

\begin{proof}[Theorem~\ref{thm_Q}]
Start from equation (\ref{eq_PeFormula}). Define $\varepsilon_n({\bf X}) = \mb{P}(\sqrtD{n}\overline{X_n} \sqrtD{n}\overline{Z_n} < -(1-2p)(\sqrtD{n}\overline{X_n})^2 \mid {\bf X}) - Q\left(\frac{(1-2p)}{\sqrt{4p(1-p)}}\left\lvert \sqrtD{n} \overline{X_n} \right\rvert \right)$. From the limiting distribution in Lemma~\ref{lemma_CLT} part (a), $\lim_{n\to\infty} \varepsilon_n({\bf X}) = 0, \forall {\bf X}$. So,
\begin{align*}
\liminf_{n\to\infty} P_e^{(n)} &= \liminf_{n\to\infty}\mb{E} [\mb{P}(\sqrtD{n}\overline{X_n} \sqrtD{n}\overline{Z_n} < \\
& \qquad \qquad \qquad -(1-2p)(\sqrtD{n}\overline{X_n})^2 \mid {\bf X}) ]\\
&= \liminf_{n\to\infty} \mb{E} \left[Q\left(\frac{(1-2p)}{\sqrt{4p(1-p)}}\left\lvert  \sqrtD{n}\overline{X_n} \right\rvert \right)\right] \\
& \qquad \qquad \qquad + \liminf_{n\to\infty} \mb{E} \left[\varepsilon_n({\bf X})\right] \\
&= \liminf_{n\to\infty} \mb{E}\left[Q\left(\frac{(1-2p)}{\sqrt{4p(1-p)}}\left\lvert \sqrtD{n} \overline{X_n} \right\rvert \right)\right].
\end{align*}
Here $\liminf_{n\to\infty} \mb{E} \left[\varepsilon_n({\bf X})\right] = \mb{E}\left[ \liminf_{n\to\infty}\varepsilon_n({\bf X}) \right] = 0$ by Lebesgue's dominated convergence theorem \cite{dudley2002real}, since $\lvert \varepsilon_n({\bf X}) \rvert \le 1$.
\end{proof}

\begin{proof}[Theorem~\ref{thm_lim}]
Mimic the proof of Theorem~\ref{thm_Q}, with all `$\liminf$' replaced by `$\lim$'.
\end{proof}

Finally, based on Theorems~\ref{thm_Hoeffd}, \ref{thm_Q} and \ref{thm_lim}, we prove Theorem~\ref{thm_tt}: whether the detection error probability tends to $0$ is exactly determined by whether $\sqrtD{n}\overline{X_n}$ asymptotically stays away from $0$, in probability. 
\begin{theorem}
\label{thm_tt}
\begin{enumerate}[(a)]
\item \label{part_tt=0}
If $\forall B > 0$, $\lim_{n\to\infty}\mb{P}(\lvert \sqrtD{n}\overline{X_n}\rvert \le B) = 0$, then $\lim_{n\to\infty}P_e^{(n)}= 0$. 
\item \label{part_tt>0}
If $\exists B > 0$ s.t. $\liminf_{n\to\infty}\mb{P}(\lvert \sqrtD{n}\overline{X_n}\rvert \le B) > 0$, then $\liminf_{n\to\infty}P_e^{(n)}> 0$.
\end{enumerate}
\end{theorem}
As a comparison to statistical physics, part (\ref{part_tt=0}) corresponds to the ferromagnetic phase where spins are mostly in one direction; part (\ref{part_tt>0}) corresponds to the paramagnetic phase where spins are nearly equal in both directions.

\begin{proof}
For part (\ref{part_tt=0}), given any $\epsilon > 0$, choose $B$ large enough such that $\exp\left(\frac{-(1-2p)^2 }{8(1-p)^2}B^2\right) \le \epsilon/2$, then choose $n$ large enough such that $\mb{P}(\lvert \sqrtD{n}\overline{X_n}\rvert \le B) \le \epsilon/2$. Theorem~\ref{thm_Hoeffd} tells that 
\begin{align*}
P_e^{(n)} &\le \mb{E}\left[\exp\left(\frac{-(1-2p)^2 }{8(1-p)^2}(\sqrtD{n}\overline{X_n})^2\right)\right] \\
&\le \mb{P}(\lvert \sqrtD{n}\overline{X_n}\rvert \le B) + \exp\left(\frac{-(1-2p)^2 }{8(1-p)^2}B^2\right) \le \epsilon.
\end{align*}
Thus $\lim_{n\to\infty}P_e^{(n)}= 0$.

For part (\ref{part_tt>0}), choose a $B > 0$, such that $\liminf_{n\to\infty}\mb{P}(\lvert \sqrtD{n}\overline{X_n}\rvert \le B) > 0$. Theorem~\ref{thm_Q} tells that 
\begin{align*}
\liminf_{n\to\infty}& P_e^{(n)} = \liminf_{n\to\infty}\mb{E} \left[Q\left(\frac{(1-2p)}{\sqrt{4p(1-p)}}\left\lvert \sqrtD{n} \overline{X_n} \right\rvert \right)\right] \\
& \ge \liminf_{n\to\infty}\mb{P}(\lvert \sqrtD{n}\overline{X_n}\rvert \le B) Q\left(\frac{(1-2p)}{\sqrt{4p(1-p)}}B \right) > 0.
\end{align*}

\end{proof}

\section{Network Examples}
\label{sec:networkExp}

Given a network topology (i.e., graph adjacency matrix $A$), theorems in the previous section characterize the asymptotic performance of the detection error probability. In this section, we will discuss two extreme network examples of Fig.~\ref{fig_model}: the chain graph (1-D Markov chain), and the complete graph. We will show that in the chain graph (as well as the empty graph as shown in Section~\ref{section_empty}), the detection error probability is asymptotically positive; in contrast, {\em in the complete graph there exists a phenomenon of phase transition}, where the detection error tends to $0$ when the connection strength is more than some critical value, while it is asymptotically positive when the connection strength is less than that critical value. {\em This is similar to the paramagnetic-ferromagnetic phase transition in statistical physics!}

\subsection{Chain Graph}
In a chain graph, each vertex is connected to two neighbors, forming a chain. For convenience, we adopt a periodical boundary condition (PBC), namely, the $n$th member is connected with the first member; nevertheless, boundary conditions asymptotically make no difference. It can be viewed as a Markov chain as: $X_{i+1} = \left\{\begin{matrix} X_{i}, &w.p.~\frac{\exp(\theta)}{\exp(\theta) + \exp(-\theta)} \\ -X_{i}, &w.p.~\frac{\exp(-\theta)}{\exp(\theta) + \exp(-\theta)} \end{matrix}\right., i = 1, \dots n,$ with $X_{n+1}$ treated as $X_1$. We will prove that the detection error probability is asymptotically positive, in Proposition~\ref{pp_Chain-Graph} below.

\begin{proposition}
\label{pp_Chain-Graph}
In the chain graph, $$\lim_{n\to\infty} P_e^{(n)} = \frac{1}{\pi}\arccot\left(\frac{(1-2p)}{\sqrt{4p(1-p)}}e^{\theta}\right) > 0.$$ 
\end{proposition}

\begin{proof}
In the chain graph, the partition function at field $b$, defined by $Z_n(\theta, b) = \sum_{{\bf x} \in \{-1, 1\}^n} \exp(\theta (x_1x_{2} + \cdots + x_{n-1}x_n + x_nx_1) + b {\bf 1}^T {\bf x})$, is \cite{pfeuty1979exact}
\begin{align}
\label{eq_chainPar}
Z_n(\theta, b) 
&= e^{n\theta}\left(\cosh b + \sqrt{\sinh^2 b + e^{-4\theta}}\right)^n \nonumber \\
& \qquad + e^{n\theta}\left(\cosh b - \sqrt{\sinh^2 b + e^{-4\theta}}\right)^n. 
\end{align}
Plug $b = 0$ into (\ref{eq_chainPar}) to get $Z_n(\theta) \doteq Z_n(\theta, 0) = (2\cosh \theta)^n + (2\sinh \theta)^n$.

The characteristic function of $\sqrtD{n} \overline{X_n}$ is
\begin{align}
\lim_{n\to\infty}&\mb{E}\left[\exp(j\omega \sqrtD{n}\overline{X_n})\right] \nonumber\\
&= \lim_{n\to\infty} \frac{1}{Z_n(\theta)} \sum\limits_{{\bf x} \in \{-1,1\}^n}\exp(\theta(x_1x_2 + \cdots + x_nx_1) \nonumber\\ 
 &\qquad + j\omega\sqrtD{n} \overline{X_n} )\nonumber\\ 
&= \lim_{n\to\infty}\frac{Z_n(\theta, \frac{j\omega}{\sqrt{n}})}{Z_n(\theta)} \nonumber\\
&= \lim_{n\to\infty} \frac{e^{n\theta}\left(\cosh \frac{j\omega}{\sqrt{n}} + \sqrt{\sinh^2 \frac{j\omega}{\sqrt{n}} + e^{-4\theta}}\right)^n}{(2\cosh \theta)^n + (2\sinh \theta)^n} \nonumber\\
 &\qquad + \lim_{n\to\infty} \frac{e^{n\theta}\left(\cosh \frac{j\omega}{\sqrt{n}} - \sqrt{\sinh^2 \frac{j\omega}{\sqrt{n}} + e^{-4\theta}}\right)^n}{(2\cosh \theta)^n + (2\sinh \theta)^n} \label{eq_TaylorChain}\\
&= \exp\left(-e^{2\theta}\omega^2/2\right).\nonumber
\end{align}
In equation (\ref{eq_TaylorChain}), we use Taylor's expansions of $\cosh\frac{j\omega}{\sqrt{n}}$ and $\sinh\frac{j\omega}{\sqrt{n}}$ to calculate the limit. As a result, $\sqrtD{n}\overline{X_n} \xrightarrow{d} N\left(0, e^{2\theta}\right)$. So, Theorem~\ref{thm_lim} tells that: $$\lim_{n\to\infty} P_e^{(n)}= \frac{1}{\pi}\arccot\left(\frac{(1-2p)}{\sqrt{4p(1-p)}}e^{\theta}\right) > 0.$$
\end{proof}
 
\subsection{Complete Graph}
In a complete graph, each pair of vertices is  connected. The corresponding Ising MRF prior (\ref{eq_CWprior}), also called the Curie-Weiss \cite{kochmanski2013curie} prior, is defined slightly differently, in that the strength is weakened to $\theta/n$, to ensure that the total strength from all neighbors of a vertex remains constant, i.e., does not grow with $n$, though the number of neighbors is $n-1$. Thus, with this standard modification, the prior is
\begin{align}
\label{eq_CWprior}
p({\bf x}) = \frac{\exp\left(\frac{\theta}{n}({\bf 1}^T{\bf x})^2\right)}{Z_n(\theta)}.
\end{align}

We are interested in the limiting distribution of $\sqrtD{n}\overline{X_n}$, with the corresponding variable called $X_{\lim}$. Intuitively, it should have the density $p(x_{\lim}) \propto \exp(\theta x_{\lim}^2)\exp(-x_{\lim}^2/2)$, since in the i.i.d. case $\sqrtD{n}\overline{X_n} \xrightarrow{d} N(0,1)$, and the Curie-Weiss prior introduces a multiplier $\exp(\theta x_{\lim}^2)$. In the case $\theta < \frac{1}{2}$, $\sqrtD{n}\overline{X_n}$ converges to a normal distribution, and therefore, by Theorem~\ref{thm_lim}, the detection error probability should be asymptotically positive; otherwise $\sqrtD{n}\overline{X_n}$ diverges so the error probability should tend to $0$. We shall prove all these formally in Proposition~\ref{pp_Complete-Graph}.
\begin{proposition}
\label{pp_Complete-Graph}
In the complete graph, 
\begin{enumerate}[(a)]
\item
if $\theta < \frac{1}{2}$, then $$\lim_{n\to\infty} P_e^{(n)} = \frac{1}{\pi}\arccot\left(\frac{(1-2p)}{\sqrt{4p(1-p)}}\frac{1}{\sqrt{1-2\theta}}\right) > 0.$$
\item
if $\theta > \frac{1}{2}$, then $\lim_{n\to\infty} P_e^{(n)}= 0$. In fact, it tends to $0$ exponentially fast: $\liminf_{n\to\infty} -\frac{1}{n}\log P_e^{(n)}> 0$.
\end{enumerate}
\end{proposition}
\begin{proof}
In this model, the partition function at field $b$, defined by $Z_n(\theta, b) = \sum_{{\bf x} \in \{-1,1\}^n}\exp(\frac{\theta}{n}({\bf 1}^T {\bf x})^2+ b {\bf 1}^T {\bf x})$, is \cite{kochmanski2013curie} 
\begin{align}
\label{eq_Curie-Weiss}
Z_n(\theta, b) = \frac{2^n}{\sqrt{\pi}}\int_{-\infty}^\infty \exp(-t^2)\cosh^n\left(2\sqrt{\frac{\theta}{n}}t+b\right)dt.
\end{align}

For part (a), we calculate the characteristic function of $\sqrtD{n}\overline{X_n}$ as
\begin{align*}
\lim_{n\to\infty}&\mb{E}\left[\exp(j\omega \sqrtD{n}\overline{X_n})\right] = \lim_{n\to\infty}\frac{Z_n(\theta, \frac{j\omega}{\sqrt n})}{Z_n(\theta)} \\
&= \frac{\int_{-\infty}^\infty \exp(-t^2)\lim_{n\to\infty}\cosh^n\left(2\sqrt{\frac{\theta}{n}}t+\frac{j\omega}{\sqrt{n}}\right)dt}{\int_{-\infty}^\infty \exp(-t^2)\lim_{n\to\infty} \cosh^n\left(2\sqrt{\frac{\theta}{n}}t\right)dt} \\
&= \frac{\int_{-\infty}^\infty \exp(-t^2)\exp\left(\frac{1}{2}(2\sqrt{\theta}t+j\omega)^2\right) dt}{\int_{-\infty}^\infty \exp(-t^2)\exp\left(\frac{1}{2}(2\sqrt{\theta}t)^2\right)dt} \\
&= \exp\left(-\frac{\omega^2}{2(1-2\theta)}\right).
\end{align*}
In the second line, the integral and the limit are exchanged by monotone convergence theorem \cite{dudley2002real}. In the third line, we use the Gaussian integral formula: $\int_{-\infty}^\infty \exp(-x^2)dx = \sqrt{\pi}$, with appropriate changes of variables. As a result, $\sqrtD{n}\overline{X_n} \xrightarrow{d} N\left(0, \frac{1}{1-2\theta}\right)$. So, Theorem~\ref{thm_lim} tells that $$\lim_{n\to\infty} P_e^{(n)}= \frac{1}{\pi}\arccot\left(\frac{(1-2p)}{\sqrt{4p(1-p)}}\frac{1}{\sqrt{1-2\theta}}\right) > 0.$$

For part (b), we only need to prove that $\liminf_{n\to\infty} -\frac{1}{n}\log P_e^{(n)}> 0$. Define $C_p = \frac{(1-2p)^2}{8(1-p)^2}$. Start from Theorem~\ref{thm_Hoeffd}:
\begin{align}
P_e^{(n)}&\le \mb{E}\left[\exp\left(-C_p(\sqrtD{n}\overline{X_n})^2\right)\right] \nonumber\\
&= \frac{1}{Z_n(\theta)}\sum_{{\bf x} \in \{-1,1\}^n}\exp\left(\frac{\theta}{n}({\bf 1}^T {\bf x})^2 - C_p(\sqrtD{n}\overline{X_n})^2\right) \nonumber\\
&= \frac{Z_n(\theta -C_p)}{Z_n(\theta)} \nonumber\\
&= \frac{\int_{-\infty}^\infty \exp(-t^2)\cosh^n\left(2\sqrt{\frac{\theta-C_p}{n}}t\right)dt}{\int_{-\infty}^\infty \exp(-t^2)\cosh^n\left(2\sqrt{\frac{\theta}{n}}t\right)dt} \nonumber\\
&= \frac{\int_{-\infty}^\infty \exp(-ns^2+n\log\cosh(2\sqrt{\theta-C_p}s))ds}{\int_{-\infty}^\infty \exp(-ns^2+n\log\cosh(2\sqrt{\theta}s))ds} \label{eq_CWint}.
\end{align}
In the fourth line, we use formula (\ref{eq_Curie-Weiss}), with $Z(\theta) \doteq Z(\theta, 0)$. In the fifth line, we change the variable to $s = t/\sqrt{n}$.

Define $f(\theta, s) = \log \cosh(2\sqrt{\theta}s) - s^2$. From equation (\ref{eq_CWint}) we obtain that:
\begin{align*}
-\frac{1}{n}\log P_e^{(n)} &\ge \frac{1}{n}\log\int_{-\infty}^\infty \exp(nf(\theta, s))ds\\
 &\qquad - \frac{1}{n}\log\int_{-\infty}^\infty \exp(nf(\theta-C_p, s))ds.
\end{align*}

By Laplace's approximation \cite{polya1997problems}, listed below as Lemma~\ref{lemma_Laplace}, whose conditions are satisfied by $f(\theta,s)$ when $\theta > \frac{1}{2}$, $\lim_{n\to\infty}\frac{1}{n}\log\int_{-\infty}^\infty \exp(nf(\theta, s))ds = \max_s f(\theta, s)$, and similarly for the second term.

It can be observed (proof is omitted here) that $\max_s f(\theta, s)$ is $0$ when $\theta \le \frac{1}{2}$ and monotonically increases with $\theta$ when $\theta > \frac{1}{2}$, as illustrated in Fig.~\ref{fig_ftheta}. Under the condition $\theta > \frac{1}{2}$, noticing that $0 < C_p \le \frac{1}{8}$ by definition, we state that $\max_s f(\theta, s) > \max_s f(\theta - C_p, s)$. In conclusion, 
\begin{align*}
\liminf_{n\to\infty} -\frac{1}{n}\log P_e^{(n)}&\ge \max_{s} f(\theta, s) - \max_{s} f(\theta-C_p, s) > 0.
\end{align*}

\begin{figure}[ht]
\centerline{\includegraphics[width=\linewidth]{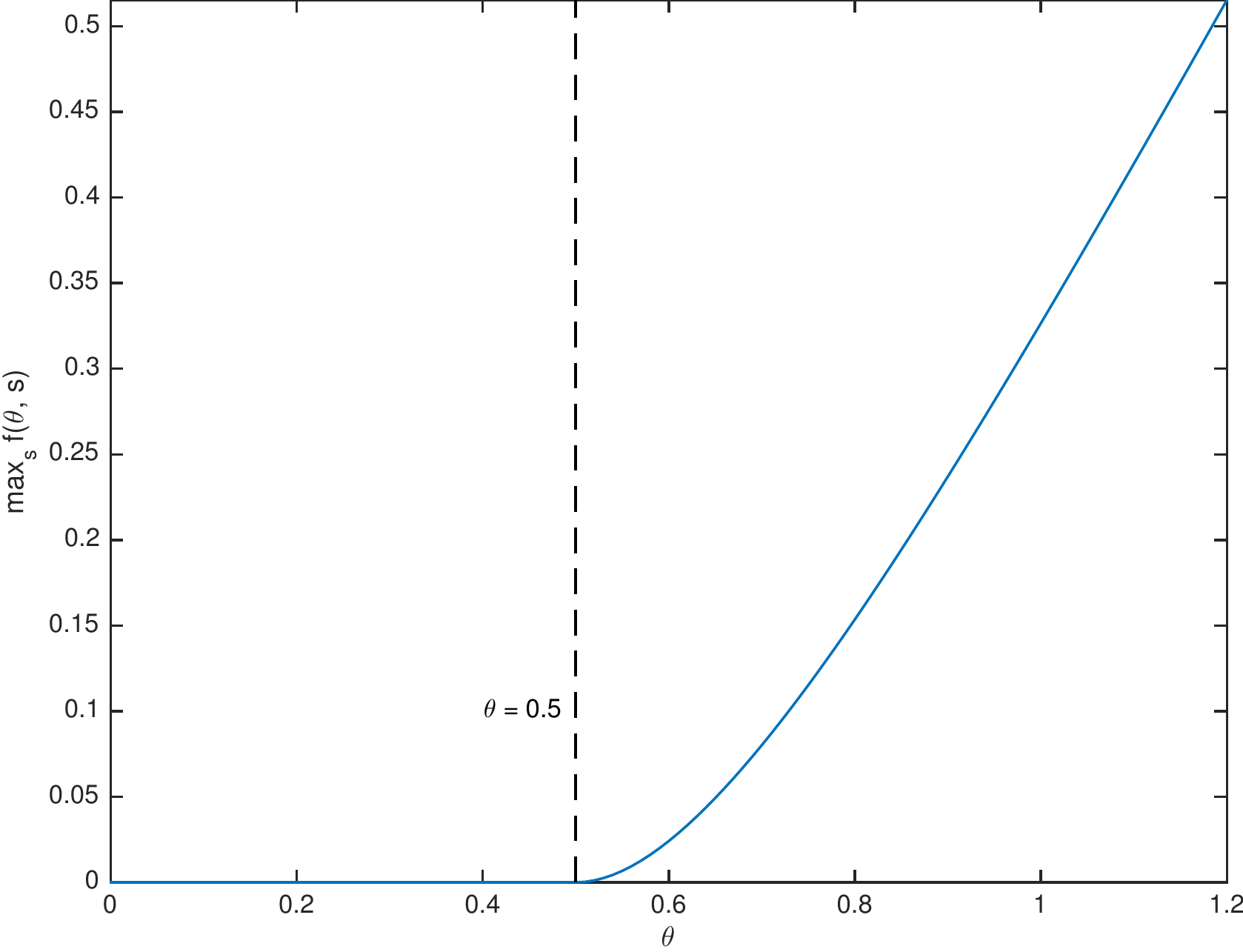}}
\caption{$\max_s f(\theta, s)$ v.s. $\theta$} 
\label{fig_ftheta}
\vspace*{-0.2cm}
\end{figure}

\begin{lemma}[Laplace's approximation]
\label{lemma_Laplace}
Given that $g(s)$ is twice differentiable, with $s^\star = \arg\max_s g(s)$, and $g^{''}(s^\star) < 0$, then
\begin{align}
\lim_{n\to\infty} \frac{1}{n}\log \int_{-\infty}^\infty \exp(ng(s))ds = \max_{s} g(s).
\end{align}
\end{lemma}
\end{proof}

\section{Numerical Results}
\label{sec:num}
In this section, we will provide numerical results for two examples to illustrate the theoretical results presented. 

First for the i.i.d. (empty graph) case, Fig.~\ref{fig_empty} illustrates the detection error probability $P_e^{(n)}$ versus user size $n$, for several cross-over probabilities $p$. The horizontal lines denote the limit in Proposition~\ref{pp_Empty-Graph}. It can be seen that $P_e^{(n)}$ tends to the positive constant given in Proposition~\ref{pp_Empty-Graph}.

\begin{figure}[ht]
\centerline{\includegraphics[width=\linewidth]{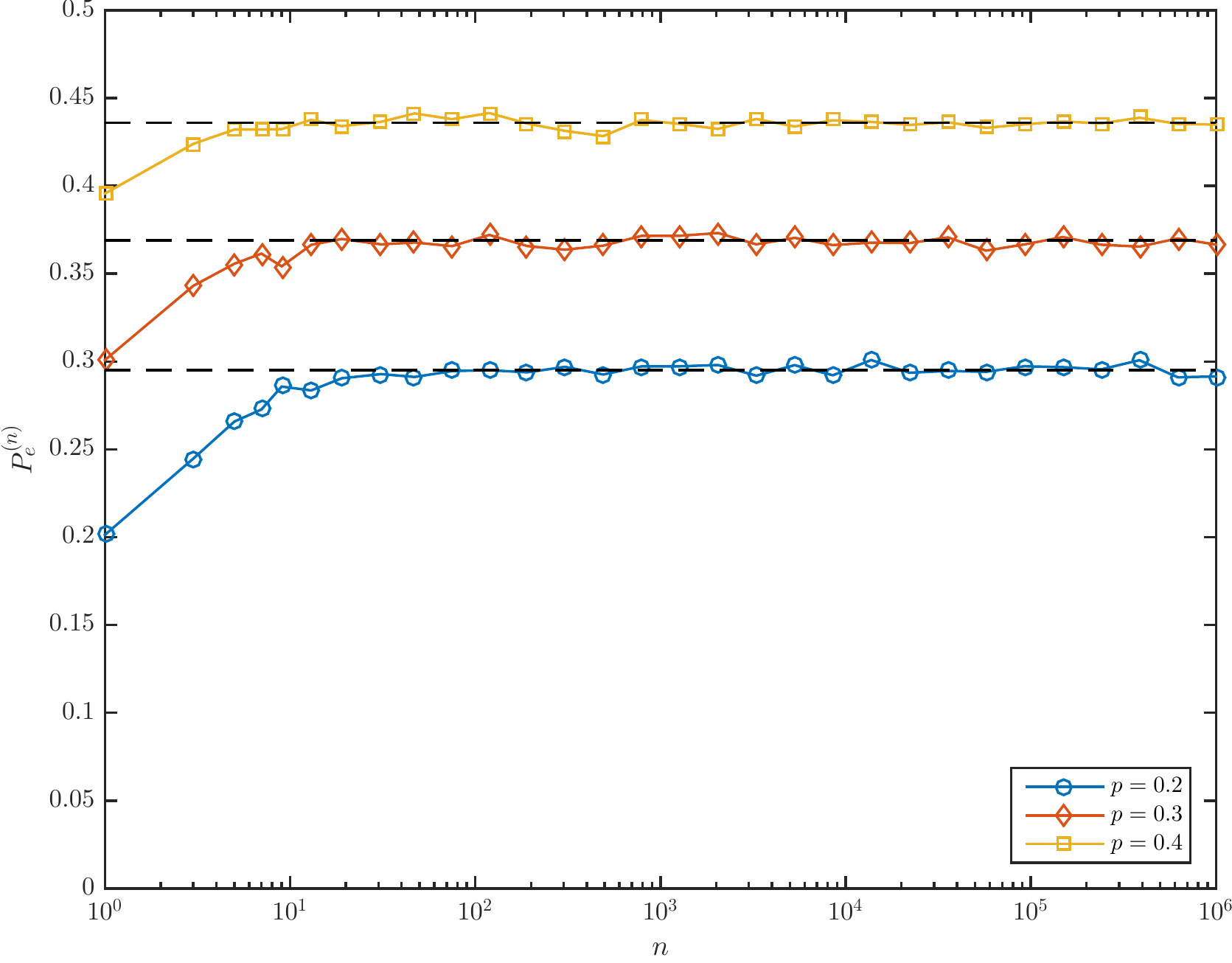}}
\caption{$P_e^{(n)}$ v.s. $n$ in i.i.d. (empty graph) case.} 
\label{fig_empty}
\vspace*{-0.2cm}
\end{figure}

Next for the complete graph case, Fig.~\ref{fig_complete} illustrates $P_e^{(n)}$ v.s. $n$ for several cross-over probabilities $p$: Fig.~\ref{fig_complete:a} with $\theta = 0.3$ corresponds to part (a) of Proposition~\ref{pp_Complete-Graph}, where the detection error probability tends to the positive constant given there; Fig.~\ref{fig_complete:b} with $\theta = 0.7$ corresponds to part (b) of Proposition~\ref{pp_Complete-Graph}, where the error probability decays to $0$ exponentially fast. 

\begin{figure}[ht]
\centering
\begin{subfigure}{.5\textwidth}
  \centering
  \includegraphics[width=\linewidth]{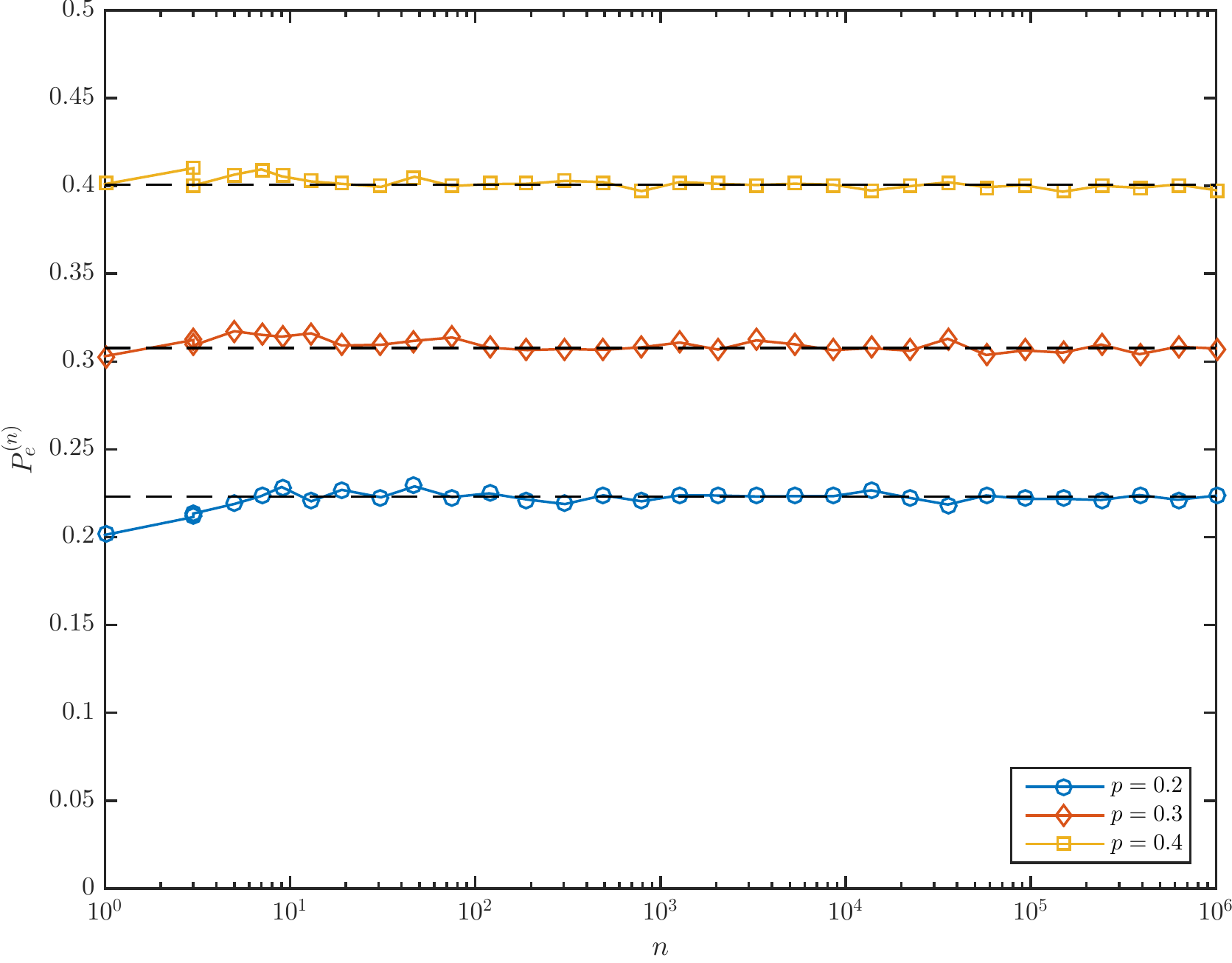}
  \caption{$\theta = 0.3$}
  \label{fig_complete:a}
\end{subfigure}
\begin{subfigure}{.5\textwidth}
  \centering
  \includegraphics[width=\linewidth]{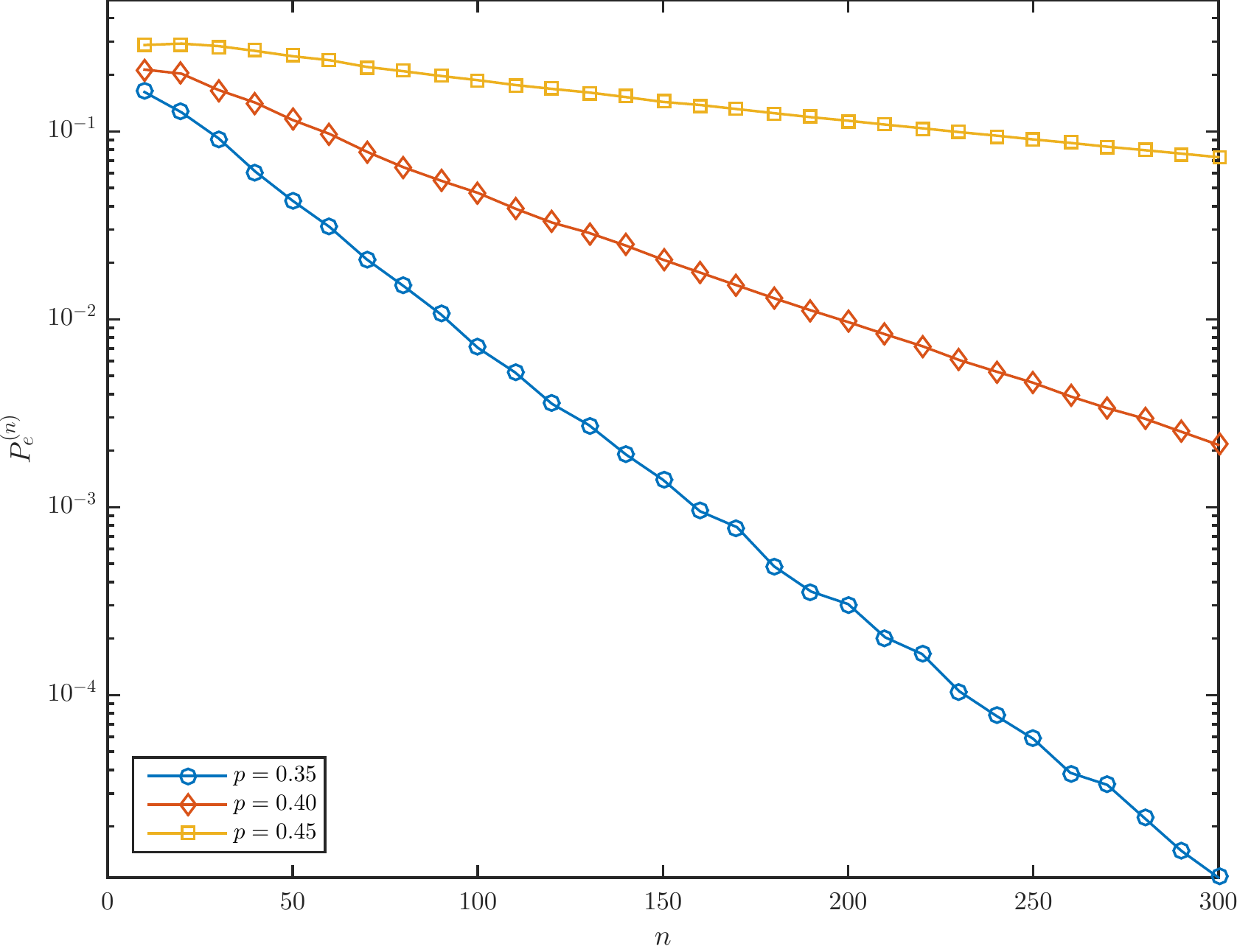}
  \caption{$\theta = 0.7$}
  \label{fig_complete:b}
\end{subfigure}
\caption{$P_e^{(n)}$ v.s. $n$ in complete graph case.} 
\label{fig_complete}
\vspace*{-0.2cm}
\end{figure}

\section{Conclusion}
\label{sec:conclusion}
In this paper, we have analyzed the asymptotic  performance of majority sentiment detection in online social networks, and revealed that the detection error probability of a majority sentiment detector is strongly related to the network connections. In the i.i.d. case, where users are not connected, the error probability never reduces to zero, regardless of how large the user base is. This result is interesting in its own right, since one would naively expect the
error probability to decay to zero with increasing number of users. 
Furthermore, by modeling the underlying social network as an Ising Markov random field prior, we discovered an interesting phenomenon of phase transition: in  the complete graph case, which is an example of a highly connected network, there exists a critical connection strength. If the strength is below 
this critical value,  the error probability is asymptotically positive, while above this critical value, the error probability tends to zero as the
number of users increases.  This phase transition seen in the complete graph case is  analogous to the critical temperature  in
 statistical physics, which separates the paramagnetic phase,
where atom magnetic spins are disordered, from the  ferromagnetic phase, where the spins are predominantly in one direction resulting in a magnet.
We remark that this phenomenon in the OSN model is not due to the type of detector used - but rather is due to the inherent similarity of opinions among users produced
by the network model.



\bibliographystyle{unsrt}
\bibliography{reflist}

\end{document}